\documentclass[12pt,oneside]{amsart}

\usepackage{amsthm,amsmath,amssymb,amsfonts,amscd,mathtools,,upgreek,url,
}
\usepackage{amsaddr}

\newcommand{\cS}{{\mathbb{S}}}
\newcommand{\tdt}{{\mathbb{D}_t}}

\newcommand{\ap}{{\upalpha}}

\newcommand{\cb}{c}

\newcommand{\ct}{\upalpha}

\newcommand{\td}{s}

\newcommand{\ud}{u_\ell,\dots}
\newcommand{\ua}{a}

\newcommand{\ub}{b}

\newcommand{\tua}{{\tilde{\ua}}}
\newcommand{\tub}{{\tilde{\ub}}}
\newcommand{\hua}{{\hat{\ua}}}
\newcommand{\hub}{{\hat{\ub}}}

\newcommand{\fn}{h}

\newcommand{\inb}{z}

\newcommand{\wv}{w}

\newcommand{\mf}{\varphi}
\newcommand{\la}{\uplambda}

\newcommand{\ff}{F}
\newcommand{\tff}{\tilde{F}}
\newcommand{\hff}{\hat{F}}

\newcommand{\GL}{\mathrm{GL}}

\newcommand{\fik}{\mathbb{C}}

\newcommand{\sm}{\mathrm{m}}

\newcommand{\lb}{\label}
\newcommand{\er}{\eqref}

\newcommand{\zp}{\mathbb{Z}_{\ge 0}}
\newcommand{\zsp}{\mathbb{Z}_{>0}}
\newcommand{\zz}{\mathbb{Z}}

\newcommand{\pd}{\partial}
\newcommand{\lt}{\mathbf{U}}

\newcommand{\mm}{\mathbf{M}}

\newcommand{\diu}{D}

\newtheorem{theorem}{Theorem}

\theoremstyle{definition}
\newtheorem{definition}{Definition}

\newtheorem{remark}{Remark}

\begin{document}

\title[Simplifications of Lax pairs and (doubly) modified integrable equations]{Simplifications of Lax pairs for differential-difference equations 
by gauge transformations and (doubly) modified integrable equations}
\date{}

\author{Sergei Igonin} 
\address{Center of Integrable Systems, P.G. Demidov Yaroslavl State University, Yaroslavl, Russia}
\email{s-igonin@yandex.ru}




\begin{abstract}

Matrix differential-difference Lax pairs
play an essential role in the theory of integrable nonlinear differential-difference equations.
We present sufficient conditions which allow one to simplify such a Lax pair by matrix gauge transformations.
Furthermore, we describe a procedure for such a simplification and present applications of it 
to constructing new integrable equations connected 
by (non-invertible) discrete substitutions of Miura type to known equations with Lax pairs.

Suppose that one has three (possibly multicomponent) equations $E$, $E_1$, $E_2$, 
a (Miura-type) discrete substitution from $E_1$ to $E$, and a discrete substitution from $E_2$ to $E_1$.
Then $E_1$ and $E_2$ can be called a modified version of~$E$
and a doubly modified version of~$E$, respectively.
We demonstrate how the above-mentioned procedure 
helps (in the considered examples) to construct modified and doubly modified versions 
of a given equation possessing a Lax pair satisfying certain conditions.

The considered examples include scalar equations of Itoh--Narita--Bogoyavlensky type
and $2$-component equations related to the Toda lattice.
We present several new integrable equations connected 
by new discrete substitutions of Miura type to known equations.

\end{abstract}

\maketitle

\section{Introduction}
\lb{secint}

This paper is devoted to a study of relations between
matrix differential-difference Lax pairs, gauge transformations, 
and (non-invertible) discrete substitutions of Miura type (also called discrete Miura-type transformations), 
which play essential roles in the theory of integrable (nonlinear) differential-difference equations.
Such equations occupy a prominent place in the modern theory of integrable systems
and are presently the subject of intensive study.
In particular, such equations arise as as discretizations of integrable PDEs 
and as chains associated with Darboux transformations of PDEs 
(see, e.g.,~\cite{HJN-book16,kmw,LWY-book22} and references therein).

As explained in Section~\ref{secprel}, 
we consider an evolution differential-difference equation 
for a vector-function $u=u(n,t)$, where
\begin{itemize}
	\item $n$ is an integer variable,
	\item $t$ is a real or complex variable,
	\item $u=u(n,t)$ takes values in the space~$\fik^\diu$ for some positive integer~$\diu$.
\end{itemize}

Matrix Lax pairs (MLPs) for such equations are defined in Definition~\ref{dmlpgt}, 
which describes also an action of the group of (matrix) gauge transformations
on the set of MLPs of a given differential-difference equation.
As explained in Definition~\ref{dmlpgt}, 
\begin{itemize}
	\item in a MLP $(\mm,\,\lt)$ 
the matrix-valued function~$\mm$ is called the \mbox{\emph{$\cS$-part}} of the MLP,
	\item two MLPs of a given equation are \emph{gauge equivalent} 
	if one MLP is obtained from the other by means of a gauge transformation.
\end{itemize}

In Section~\ref{secsimpl} we present sufficient conditions 
which allow one to simplify a given MLP by gauge transformations 
and describe a procedure for such a simplification.
The words 
$$
\text{``to simplify a given MLP by gauge transformations''}
$$
mean that, applying suitable gauge transformations to a given MLP 
satisfying certain conditions, 
we eliminate the dependence on~$u_k$ for some values of~$k\in\zz$ 
in the~\mbox{$\cS$-part} of the MLP.
(See Remark~\ref{rsgt} for more details.)
Here $u_k$ denotes the vector-function~$u_k(n,t)=u(n+k,t)$.

It is well known that Miura-type transformations (MTs) belong 
to the most important tools in the theories of partial differential, difference 
and differential-difference equations~\cite{miura68,mss91,yam2006,garif2018,GarYam2019,gram11,meshk2008,LWY-book22,sokolov88,startsev01,suris03,yam94}.
MTs for partial differential equations are also called 
differential substitutions~\cite{sokolov88,startsev01}, 
while MTs for differential-difference equations are sometimes called discrete substitutions~\cite{yam94}.

Suppose that we have three (possibly multicomponent) equations $E$, $E_1$, $E_2$, 
a MT from~$E_1$ to~$E$, and a MT from~$E_2$ to~$E_1$.
Then one can say that
\begin{itemize}
	\item $E_1$ is a modified version of~$E$,
	\item $E_2$ is a modified version of~$E_1$ and is a doubly modified version of~$E$.
\end{itemize}
Precise definitions are given in Section~\ref{secprel}.

When one classifies a certain class of integrable equations with two independent variables, 
one usually gets several basic equations 
such that all the other equations from the class under consideration can be obtained from 
the basic ones by applying MTs and are modified or doubly modified versions of the basic ones
(see, e.g.,~\cite{mss91,yam2006,garif2018,GarYam2019,gram11,meshk2008,LWY-book22}).
Therefore, it is important to develop methods to construct MTs as well as 
modified and doubly modified versions for a given integrable equation.

In Sections~\ref{secind},~\ref{seceqtod} we demonstrate how 
the above-mentioned procedure to simplify MLPs helps 
(in the considered examples) to construct modified and doubly modified versions 
of a given equation possessing a MLP satisfying certain conditions.
This allows one to derive new integrable equations connected by new MTs to known equations.

The considered examples include scalar equations of Itoh--Narita--Bogoyavlensky type 
(in Section~\ref{secind})
and $2$-component equations related to the Toda lattice (in Section~\ref{seceqtod}).
The obtained integrable equations~\er{vtlong}, \er{mtoda}, \er{longtoda} and 
MTs \er{dslong}, \er{ds4}, \er{dstoda}, \er{dsltoda} are new, to our knowledge.

The paper is organized as follows.
Section~\ref{secprel} contains preliminaries, 
including conventions, notation, basic notions, and terminology.
Section~\ref{secsimpl} describes sufficient conditions which allow one 
to simplify a MLP by gauge transformations and a procedure for such a simplification.
Sections~\ref{secind},~\ref{seceqtod} contain examples of applications of this procedure, 
as outlined above.
Section~\ref{sconcl} summarizes the results of this paper, 
compares them with the results of~\cite{BIg2016},
and suggests some open problems for further research.

\section{Preliminaries}
\lb{secprel}


In this paper we use the following notation and conventions.
\begin{itemize}
	\item MLP = matrix Lax pair. 
	\item MT = Miura-type transformation. We study discrete MTs, which are also called discrete substitutions.
	\item $\zp$ and $\zsp$ are the sets of nonnegative and positive integers, respectively.
	\item Each considered function 
is assumed to be analytic on its domain of definition.
In particular, this holds for meromorphic functions. By our convention, the poles of
a meromorphic function do not belong to its domain of definition, thus
such a function is analytic on its domain of definition.
\item Unless otherwise specified, scalar variables and functions are supposed to be 
$\fik$-valued. 
\item $n$ is an integer variable, while $t$ is a real or complex variable.
\item For any function $w=w(n,t)$ and each $\ell\in\zz$ 
we denote by~$w_\ell$ the function $w_\ell(n,t)=w(n+\ell,t)$.
In particular, $w_0=w$.
\item Matrix-valued functions are sometimes called simply matrices.
\end{itemize}

Fix $\diu\in\zsp$. Let $\ua,\ub\in\zz$ such that $\ua\le\ub$.
We consider an evolution differential-difference equation of the form
\begin{gather}
\lb{sdde}
u_t=\ff(u_\ua,u_{\ua+1},\dots,u_\ub)
\end{gather}
for a $\diu$-component vector-function $u=\big(u^1(n,t),\dots,u^{\diu}(n,t)\big)$,
where 
\begin{itemize}
	\item $\ff$ is a $\diu$-component vector-function $\ff=(\ff^1,\dots,\ff^{\diu})$,
	\item $u_t=\dfrac{\pd}{\pd t}(u)$ and $u_\ell=u_\ell(n,t)=u(n+\ell,t)$ for $\ell\in\zz$.
\end{itemize}
One has $u_\ell=(u^1_\ell,\dots,u^\diu_\ell)$, where $u_\ell^\xi(n,t)=u^\xi(n+\ell,t)$ 
for $\xi=1,\dots,\diu$.

Equation \eqref{sdde} is equivalent to the following infinite collection of differential equations
$$
\frac{\pd}{\pd t}\big(u(n,t)\big)=\ff\big(u(n+\ua,t),u(n+\ua+1,t),\dots,u(n+\ub,t)\big),\qquad\quad n\in\zz.
$$
In components equation~\eqref{sdde} reads
\begin{gather}
\lb{msdde}
\frac{\pd}{\pd t}\big(u^i\big)=\ff^i(u_\ua^\xi,u_{\ua+1}^\xi,\dots,u_\ub^\xi),\qquad\quad
i=1,\dots,\diu,
\end{gather}
which implies 
\begin{gather}
\lb{uildde}
\frac{\pd}{\pd t}\big(u^i_\ell\big)=\ff^i(u_{\ua+\ell}^\xi,u_{\ua+1+\ell}^\xi,\dots,u_{\ub+\ell}^\xi),\qquad\quad
i=1,\dots,\diu,\qquad \ell\in\zz.
\end{gather}

We use the formal theory of differential-difference equations, where one regards 
\begin{gather}
\lb{uldiu}
u_\ell=(u^1_\ell,\dots,u^\diu_\ell),\qquad\quad \ell\in\zz,
\end{gather}
as independent quantities, which are called \emph{dynamical variables}.
We consider functions of the dynamical variables~\er{uldiu}.
In this paper, the notation of the type $\fn=\fn(\ud)$ means 
that a function~$\fn$ depends on a finite 
number of the variables $u_\ell^{\xi}$ for $\ell\in\zz$ and $\xi=1,\dots,\diu$.

The notation of the type 
$\fn=\fn(u_\alpha,\dots,u_\beta)$ or $\fn=\fn(u_\alpha,u_{\alpha+1},\dots,u_\beta)$ 
for some integers $\alpha\le\beta$ means that $\fn$ 
may depend on $u_\ell^{\xi}$ for $\ell=\alpha,\dots,\beta$ 
and $\xi=1,\dots,\diu$.

We denote by~$\cS$ the \emph{shift operator} with respect to the variable~$n$.
For any function $g=g(n,t)$ one has the function~$\cS(g)$ 
such that $\cS(g)(n,t)=g(n+1,t)$.
For each $k\in\zz$, we have the $k$th power~$\cS^k$ 
of the operator~$\cS$ and the formula $\cS^k(g)(n,t)=g(n+k,t)$.

Since $u_\ell$ corresponds to $u(n+\ell,t)$, the operator~$\cS$ 
and its powers~$\cS^k$ for $k\in\zz$ act on functions of~$u_\ell$ as follows
\begin{gather}
\lb{csuf}
\cS(u_\ell)=u_{\ell+1},\qquad\cS^k(u_\ell)=u_{\ell+k},\qquad
\cS^k\big(\fn(u_\ell,\dots)\big)=\fn(\cS^k(u_{\ell}),\dots).
\end{gather}
That is, applying $\cS^k$ to a function $\fn=\fn(\ud)$, 
we replace~$u_\ell^{\xi}$ by~$u_{\ell+k}^{\xi}$ in~$\fn$ for all $\ell,\xi$.

The \emph{total derivative operator}~$\tdt$ corresponding to~\eqref{sdde} 
acts on functions of the variables~$u_\ell=(u^1_\ell,\dots,u^\diu_\ell)$ as follows
\begin{gather}
\lb{dtfu}
\tdt\big(\fn(\ud)\big)=\sum_{\substack{\ell\in\zz,\\ 
i=1,\dots,\diu}}
\cS^\ell(\ff^i)\cdot\frac{\pd \fn}{\pd u_\ell^i},
\end{gather}
where $\ff^i$ are the components of the vector-function 
$\ff=(\ff^1,\dots,\ff^\diu)$ from~\eqref{sdde}.
Formula~\er{dtfu} is explained by equations~\er{uildde} 
and the chain rule for the derivative with respect to~$t$.
Formula~\er{dtfu} implies 
the relation $\tdt\big(\cS(g)\big)=\cS\big(\tdt(g)\big)$ for any function $g=g(\ud)$.


\begin{definition}
\lb{dmlpgt}
Let $\sm\in\zsp$. 
Let $\mm=\mm(\ud,\la)$ and $\lt=\lt(\ud,\la)$ be $\sm\times\sm$ matrix-valued functions
depending on the variables $u_\ell$ and a complex parameter $\la$.
Suppose that $\mm$ is invertible (i.e., $\mm$ takes values in the group~$\GL_\sm(\fik)$ 
of invertible $\sm\times\sm$ matrices) and one has
\begin{gather}
\lb{lr}
\tdt(\mm)=\cS(\lt)\mm-\mm\lt,
\end{gather}
where $\tdt$ is given by~\eqref{dtfu}.
Then the pair $(\mm,\,\lt)$ is called a \emph{matrix Lax pair} (MLP) for equation~\eqref{sdde}.
Equation~\er{lr} implies that the auxiliary linear system 
\begin{gather}
\lb{syspsi}
\begin{aligned}
\cS(\Phi)&=\mm\Phi,\\
\frac{\pd}{\pd t}(\Phi)&=\lt\Phi
\end{aligned}
\end{gather}
is compatible modulo~\eqref{sdde}. 
Here $\Phi=\Phi(n,t,\la)$ is an invertible $\sm\times\sm$ matrix-valued function.

We say that the matrix $\mm=\mm(\ud,\la)$ is the \emph{$\cS$-part} of the MLP $(\mm,\,\lt)$.

Then for any invertible $\sm\times\sm$ matrix  
$\mathbf{g}=\mathbf{g}(\ud,\la)$ the matrices
\begin{gather}
\lb{ggtmtu}
\hat \mm=\cS(\mathbf{g})\cdot \mm\cdot\mathbf{g}^{-1},\qquad\quad
\hat\lt=\tdt(\mathbf{g})\cdot\mathbf{g}^{-1}+
\mathbf{g}\cdot\lt\cdot\mathbf{g}^{-1}
\end{gather}
form a MLP for equation~\eqref{sdde} as well.
The MLP $(\hat \mm,\,\hat\lt)$ is \emph{gauge equivalent} 
to the MLP $(\mm,\,\lt)$ and is obtained from $(\mm,\,\lt)$ 
by means of the \emph{gauge transformation}~$\mathbf{g}$.

Such gauge transformations~$\mathbf{g}$ form a group with respect to the multiplication of matrices.
Formulas~\er{ggtmtu} determine an action of the group of gauge transformations
on the set of MLPs of a given equation~\eqref{sdde}.
\end{definition}

\begin{remark}
\lb{reminv}
Definition~\ref{dmlpgt} says that, in any MLP $\big(\mm(\ud,\la),\,\lt(\ud,\la)\big)$,
the matrix~$\mm(\ud,\la)$ is required to be invertible.
However, in some examples of MLPs $\big(\mm(\ud,\la),\,\lt(\ud,\la)\big)$, 
it may happen that the matrix~$\mm(\ud,\la)$ is invertible for almost all (but not all) 
values of~$u_\ell$ and~$\la$.
The exceptional points $(\ud,\la)$ where the matrix~$\mm(\ud,\la)$ 
is singular are excluded from consideration in this paper.

More precisely, the situation is as follows.
There are integers $p\le q$ such that 
$$
\mm(\ud,\la)=\mm(u_p,u_{p+1},\dots,u_q,\la).
$$
Recall that $u_{\ell}=(u^1_{\ell},\dots,u^\diu_{\ell})$ is a $\diu$-dimensional vector 
for each $\ell\in\mathbb{Z}$, and $\la$ is a scalar parameter.
We assume that the $\sm\times\sm$ matrix-valued function $\mm(u_p,u_{p+1},\dots,u_q,\la)$
is defined and is invertible on an open subset~$\mathbb{W}$ 
of the space~$\fik^{(q-p+1)\diu+1}$ with the following coordinates
\begin{gather}
\notag
u_\ell^{\xi},\qquad\la,\qquad\quad \ell=p,p+1,\dots,q,\qquad\xi=1,\dots,\diu.
\end{gather}
\end{remark}

Similarly to~\eqref{sdde}, let $\tua,\tub\in\zz$, $\,\tua\le\tub$, 
and consider another differential-difference equation
\begin{gather}
\lb{vdde}
v_t=\tff(v_\tua,v_{\tua+1},\dots,v_\tub)
\end{gather}
for a $\diu$-component vector-function $v=\big(v^1(n,t),\dots,v^{\diu}(n,t)\big)$.
One has the corresponding dynamical variables $v_\ell=(v^1_\ell,\dots,v^\diu_\ell)$, $\ell\in\zz$.
Similarly to~\eqref{csuf},~\eqref{dtfu}, the operators $\cS$ and $\tdt$ act on functions 
of the variables~$v_\ell=(v^1_\ell,\dots,v^\diu_\ell)$ as follows
\begin{gather*}
\cS(v_\ell)=v_{\ell+1},\qquad\cS^k(v_\ell)=v_{\ell+k},\qquad
\cS^k\big(\fn(v_\ell,\dots)\big)=\fn(\cS^k(v_{\ell}),\dots),\qquad k\in\zz,\\
\tdt\big(\fn(v_\ell,\dots)\big)=\sum_{\ell,i}\cS^\ell(\tff^i)\cdot\frac{\pd \fn}{\pd v_\ell^i},
\end{gather*}
where $\tff^i$ are the components of the vector-function 
$\tff=(\tff^1,\dots,\tff^\diu)$ from~\eqref{vdde}.

\begin{definition}
\lb{defmtt}
A \emph{Miura-type transformation} (MT) from equation~\eqref{vdde} 
to equation~\eqref{sdde} is given by a formula of the type
\begin{gather}
\lb{uvf}
u=\mf(v_\ell,\dots)
\end{gather}
(where a $\diu$-component vector-function $\mf$ depends on 
a finite number of the variables $v_\ell=(v^1_\ell,\dots,v^\diu_\ell)$, $\ell\in\zz$,)
such that if $v=v(n,t)$ obeys~\eqref{vdde} then $u=u(n,t)$ 
given by~\eqref{uvf} obeys~\eqref{sdde}.

More precisely, in order to be a MT from~\eqref{vdde} to~\eqref{sdde},
formula~\eqref{uvf} must satisfy equations~\er{tdtmf} explained below. 
In components \eqref{uvf} reads 
\begin{gather}
\lb{uimfi}
u^i=\mf^i(v^\xi_\ell,\dots),\qquad\quad i=1,\dots,\diu,
\end{gather}
where $\mf^i$ are the components of the vector-function $\mf=(\mf^1,\dots,\mf^{\diu})$ from~\er{uvf}.
If we substitute the right-hand side of~\eqref{uimfi} in place of $u^i$ in~\eqref{msdde}, we obtain 
\begin{gather}
\lb{tdtmf}
\tdt\big(\mf^i(v^\xi_\ell,\dots)\big)=
\ff^i\big(\cS^\ua(\mf^\xi),\cS^{\ua+1}(\mf^\xi),\dots,\cS^\ub(\mf^\xi)\big),\qquad\quad
i=1,\dots,\diu,
\end{gather}
which must be an identity in the variables $v^\xi_\ell$.

Let $s,m\in\zz$, $\,s\le m$, be such that the function~$\mf$ in~\eqref{uvf} 
is of the form $\mf=\mf(v_s,v_{s+1},\dots,v_m)$ 
(i.e., $\mf$ may depend only on $v_s,v_{s+1},\dots,v_m$)
and depends nontrivially on $v_s$, $v_m$.
Then the number $m-s$ is said to be the \emph{order} of the MT~\eqref{uvf}.

If \er{vdde} and~\eqref{sdde} are connected by a MT~\er{uvf}, 
then equation~\er{vdde} can be called a \emph{modified version} of equation~\eqref{sdde}.
This notion is nontrivial when the transformation~\er{uvf} is non-invertible, 
which includes all the examples considered in this paper.
\end{definition}

Furthermore, let $\hua,\hub\in\zz$, $\,\hua\le\hub$, 
and consider another differential-difference equation
\begin{gather}
\lb{wdde}
w_t=\hff(w_\hua,w_{\hua+1},\dots,w_\hub)
\end{gather}
for a $\diu$-component vector-function $w=\big(w^1(n,t),\dots,w^{\diu}(n,t)\big)$.
Suppose that there is a MT
\begin{gather}
\lb{uwds}
v=\psi(w_\ell,\dots)
\end{gather}
from equation~\er{wdde} to equation~\er{vdde}.
In this case, since \er{wdde} is connected to~\eqref{sdde} 
by the composition of the two MTs~\er{uwds} and~\er{uvf}, 
one can say that equation~\er{wdde} is a \emph{doubly modified version} of equation~\eqref{sdde}.

\begin{remark}
\lb{rmcds}
The paper~\cite{BIg2016} presents a method to derive 
MTs from matrix Lax pairs~$(\mm,\,\lt)$
in the case when $\mm=\mm(u_0,\la)$ depends only on~$u_0$,~$\la$ 
and satisfies certain conditions.
Some ideas behind the method of~\cite{BIg2016} are inspired
by a result of V.G.~Drinfeld and V.V.~Sokolov on MTs
for the partial differential Korteweg--de Vries equation~\cite{drin-sok85}.

In~\cite{BIg2016,kmw}, matrix Lax pairs 
for differential-difference equations are called Darboux--Lax representations, 
since many of them arise from Darboux transformations of PDEs (see, e.g.,~\cite{kmw}).
\end{remark}

\section{Simplifications of matrix Lax pairs by gauge transformations}
\lb{secsimpl}

According to Definition~\ref{dmlpgt}, 
in a MLP~$(\mm,\,\lt)$ the matrix-valued function~$\mm$ may depend on 
any finite number of the variables~$u_{\ell}=(u^1_{\ell},\dots,u^\diu_{\ell})$, $\ell\in\zz$, 
and a parameter~$\la$. For any fixed integers $\ell_1,\dots,\ell_\diu$, we can relabel 
\begin{gather}
\lb{relab}
u^1:=u^1_{\ell_1},\quad\dots\quad,\quad u^{\diu}:=u^{\diu}_{\ell_\diu}.
\end{gather}
Relabeling~\er{relab} means that in equation~\er{sdde} 
we make the following invertible change of variables 
\begin{gather}
\notag
u^1(n,t)\,\mapsto\,u^1(n+\ell_1,t),\quad\dots\quad,\quad 
u^{\diu}(n,t)\,\mapsto\,u^{\diu}(n+\ell_\diu,t).
\end{gather}
After a suitable relabeling of this type, we can assume that $\mm$ 
is of the form $\mm=\mm(u_0,\dots,u_k,\la)$ for some $k\in\zp$
(that is, we can assume that $\mm$ does not depend on~$u_s$ for negative integers~$s$).

As discussed in Remark~\ref{rsgt} below 
(using the terminology introduced in Definitions~\ref{dord},~\ref{dsim}),
Theorem~\ref{thmuk} gives sufficient conditions 
which allow one to simplify a MLP by gauge transformations.

\begin{theorem}
\label{thmuk}
Let $\sm,k\in\zsp$. Consider an $\sm\times\sm$ matrix-valued function~$\mm=\mm(u_0,u_1,\dots,u_k,\la)$, 
where $u_{\ell}=(u^1_{\ell},\dots,u^\diu_{\ell})$ for any $\ell\in\zp$. 
Suppose that 
\begin{gather}
\lb{pdm}
\forall\, i,j=1,\dots,\diu\qquad\quad
\frac{\pd}{\pd u^i_0}\Big(\frac{\pd}{\pd u^j_k}
\big(\mm(u_0,u_1,\dots,u_k,\la)\big)\cdot \mm(u_0,u_1,\dots,u_k,\la)^{-1}\Big)=0.
\end{gather}

Then there is an $\sm\times\sm$ matrix gauge 
transformation~$\mathbf{g}=\mathbf{g}(u_0,\dots,u_{k-1},\la)$ 
such that the matrix-valued function 
\begin{gather}
\lb{hatm}
\hat \mm=\cS(\mathbf{g})\cdot \mm\cdot\mathbf{g}^{-1}
\end{gather}
is of the form $\hat \mm=\hat \mm(u_0,\dots,u_{k-1},\la)$. 
Thus, applying this gauge transformation, we eliminate the dependence 
on~$u_k=(u^1_k,\dots,u^\diu_k)$.
\end{theorem}
\begin{proof}
Recall that $u_{\ell}$ are $\diu$-dimensional vectors and $\la$ is a scalar parameter.
According to Remark~\ref{reminv}, the matrix-valued function
$\mm(u_0,u_1,\dots,u_k,\la)$ is defined and is invertible on an open subset~$\mathbb{W}$ 
of~$\fik^{(k+1)\diu+1}$.
Without loss of generality, we can assume the set~$\mathbb{W}$ to be connected.
(If $\mathbb{W}$ is not connected, then one can repeat the arguments presented below 
for each connected component of~$\mathbb{W}$ separately.)

Fix a constant vector $a_0=(a^1_0,\dots,a^\diu_0)\in\fik^\diu$ and substitute $u_0=a_0$ 
in $\mm(u_0,u_1,\dots,u_k,\la)$. 
The vector $a_0$ is chosen so that $\mm(a_0,u_1,\dots,u_k,\la)$ is defined and is invertible.

We set 
\begin{gather}
\lb{guum}
\mathbf{g}(u_0,\dots,u_{k-1},\la)=\cS^{-1}\Big(\big(\mm(a_0,u_1,\dots,u_k,\la)\big)^{-1}\Big).
\end{gather}
Then, since $\cS(\mathbf{g})(u_1,\dots,u_k,\la)=\mm(a_0,u_1,\dots,u_k,\la)^{-1}$, 
for the matrix-valued function~\er{hatm} one has
\begin{multline}
\lb{hatmu}
\hat \mm(u_0,u_1,\dots,u_k,\la)=\cS(\mathbf{g})(u_1,\dots,u_k,\la)\cdot \mm(u_0,u_1,\dots,u_k,\la)
\cdot\mathbf{g}^{-1}(u_0,\dots,u_{k-1},\la)=\\
=\mm(a_0,u_1,\dots,u_k,\la)^{-1}\cdot \mm(u_0,u_1,\dots,u_k,\la)
\cdot\mathbf{g}^{-1}(u_0,\dots,u_{k-1},\la).
\end{multline}
Let $j\in\{1,\dots,\diu\}$. Set $\mathbf{G}=\mathbf{g}^{-1}(u_0,\dots,u_{k-1},\la)$.
Since $\dfrac{\pd}{\pd u^j_k}(\mathbf{G})=\dfrac{\pd}{\pd u^j_k}\big(\mathbf{g}^{-1}(u_0,\dots,u_{k-1},\la)\big)=0$, from~\er{hatmu} we get
\begin{multline}
\lb{pdhatm}
\frac{\pd}{\pd u^j_k}\big(\hat \mm(u_0,u_1,\dots,u_k,\la)\big)
=\frac{\pd}{\pd u^j_k}\Big(\mm(a_0,u_1,\dots,u_k,\la)^{-1}\cdot \mm(u_0,u_1,\dots,u_k,\la)\Big)
\cdot\mathbf{G}=\\
=\Big(\frac{\pd}{\pd u^j_k}\big(\mm(a_0,u_1,\dots,u_k,\la)^{-1}\big)\cdot \mm(u_0,u_1,\dots,u_k,\la)+\\
+\mm(a_0,u_1,\dots,u_k,\la)^{-1}\cdot\frac{\pd}{\pd u^j_k}\big(\mm(u_0,u_1,\dots,u_k,\la)\big)\Big)
\cdot\mathbf{G}=\\
=\Big(-\mm(a_0,u_1,\dots,u_k,\la)^{-1}\cdot\frac{\pd}{\pd u^j_k}\big(\mm(a_0,u_1,\dots,u_k,\la\big)\cdot 
\mm(a_0,u_1,\dots,u_k,\la)^{-1}\cdot \mm(u_0,u_1,\dots,u_k,\la)+\\
+\mm(a_0,u_1,\dots,u_k,\la)^{-1}\cdot\frac{\pd}{\pd u^j_k}\big(\mm(u_0,u_1,\dots,u_k,\la)\big)\Big)
\cdot\mathbf{G}=\\
=\mm(a_0,u_1,\dots,u_k,\la)^{-1}\cdot\Big(-\frac{\pd}{\pd u^j_k}\big(\mm(a_0,u_1,\dots,u_k,\la)\big)\cdot
\mm(a_0,u_1,\dots,u_k,\la)^{-1}+\\
+\frac{\pd}{\pd u^j_k}\big(\mm(u_0,u_1,\dots,u_k,\la)\big)\cdot
\mm(u_0,u_1,\dots,u_k,\la)^{-1}\Big)\cdot \mm(u_0,u_1,\dots,u_k,\la)\cdot\mathbf{G}=\\
=\mm(a_0,u_1,\dots,u_k,\la)^{-1}\cdot L^j(u_0,u_1,\dots,u_k,\la)\cdot \mm(u_0,u_1,\dots,u_k,\la)\cdot\mathbf{G},
\end{multline}
where 
\begin{multline}
\lb{Lu}
L^j(u_0,u_1,\dots,u_k,\la)=-\frac{\pd}{\pd u^j_k}\big(\mm(a_0,u_1,\dots,u_k,\la)\big)\cdot
\mm(a_0,u_1,\dots,u_k,\la)^{-1}+\\
+\frac{\pd}{\pd u^j_k}\big(\mm(u_0,u_1,\dots,u_k,\la)\big)\cdot \mm(u_0,u_1,\dots,u_k,\la)^{-1}.
\end{multline}
From~\er{Lu} and~\er{pdm} it follows that
\begin{gather}
\lb{pdL}
\begin{gathered}
\frac{\pd}{\pd u^i_0}\big(L^j(u_0,u_1,\dots,u_k,\la)\big)=
\frac{\pd}{\pd u^i_0}\Big(\frac{\pd}{\pd u^j_k}
\big(\mm(u_0,u_1,\dots,u_k,\la)\big)\cdot \mm(u_0,u_1,\dots,u_k,\la)^{-1}\Big)=0\\
\forall\, i=1,\dots,\diu,
\end{gathered}\\
\lb{La0}
L^j(a_0,u_1,\dots,u_k,\la)=0.
\end{gather}
Equations~\er{pdL},~\er{La0} imply that the matrix-valued function
$L^j(u_0,u_1,\dots,u_k,\la)$ is identically zero.

Substituting $L^j(u_0,u_1,\dots,u_k,\la)=0$ in~\er{pdhatm}, one obtains
\begin{gather}
\notag
\frac{\pd}{\pd u^j_k}\big(\hat \mm(u_0,u_1,\dots,u_k,\la)\big)=0\qquad\quad
\forall\, j=1,\dots,\diu.
\end{gather}
Therefore, the matrix-valued function~\er{hatm} is of the form $\hat \mm=\hat \mm(u_0,\dots,u_{k-1},\la)$. 
\end{proof}

Motivated by the result of Theorem~\ref{thmuk},
we give Definitions~\ref{dord},~\ref{dsim} and discuss this result in Remark~\ref{rsgt} below.
\begin{definition}
\lb{dord}
Let $p\in\zp$. A MLP $\big(\mm(\ud,\la),\,\lt(\ud,\la)\big)$ is said to be \emph{of~order~$\le p$} 
if its \mbox{$\cS$-part}~$\mm(\ud,\la)$ is of the form $\mm=\mm(u_0,\dots,u_{p},\la)$. 
In other words, a MLP $(\mm,\,\lt)$ is of~order~$\le p$ if
\begin{gather}
\notag
\forall\,i=1,\dots,\diu\qquad\forall\,\ell_1<0\qquad\forall\,\ell_2>p\qquad
\frac{\pd}{\pd u^i_{\ell_1}}(\mm)=0,\qquad\frac{\pd}{\pd u^i_{\ell_2}}(\mm)=0.
\end{gather}
\end{definition}
\begin{definition}
\lb{dsim}
Let $q\in\zsp$. A MLP $\big(\mm(\ud,\la),\,\lt(\ud,\la)\big)$ is said to be \emph{$q$-flat} 
if it is of~order~${\le q}$ and its \mbox{$\cS$-part} $\mm=\mm(u_0,\dots,u_{q},\la)$ satisfies
\begin{gather}
\lb{pdjq}
\forall\, i,j=1,\dots,\diu\qquad\quad
\frac{\pd}{\pd u^i_0}\Big(\frac{\pd}{\pd u^j_q}(\mm)\cdot \mm^{-1}\Big)=0.
\end{gather}
Note that for $q=k$ equation~\er{pdjq} coincides with equation~\er{pdm} from Theorem~\ref{thmuk}.
\end{definition}
\begin{remark}
\lb{rsgt}
Theorem~\ref{thmuk} implies the following. Let $k\in\zsp$.
If a MLP is $k$-flat, then, applying to it a suitable gauge transformation, 
we obtain a gauge equivalent MLP of~order~$\le k\!-\!1$. 
Similarly, if the obtained gauge equivalent MLP is $(k\!-\!1)$-flat (when $k\ge 2$), then, 
applying to it another gauge transformation, one gets a gauge equivalent MLP of~order~$\le k\!-\!2$.
If this MLP is $(k\!-\!2)$-flat (when $k\ge 3$), one can use the same procedure again, recursively.
We repeat this procedure as many times as possible, and the result is discussed below.

Thus, starting from a $k$-flat MLP $\big(\mm,\,\lt)$ with 
the \mbox{$\cS$-part} $\mm=\mm(u_0,\dots,u_{k},\la)$ 
and applying several gauge transformations recursively, 
we can derive a gauge equivalent MLP $\big(\check{\mm},\,\check{\lt})$ of~order~$\le p$ 
with some nonnegative integer~$p<k$ and the \mbox{$\cS$-part} 
$\check{\mm}=\check{\mm}(u_0,\dots,u_{p},\la)$. 
Therefore, the described procedure allows us to simplify the initial MLP $\big(\mm,\,\lt)$ 
in the sense that we eliminate the dependence on~$u_{p+1},\dots,u_k$ 
in the \mbox{$\cS$-part}.

Thus, Theorem~\ref{thmuk} gives sufficient conditions which allow one 
to simplify (in the above-mentioned sense) a MLP by applying gauge transformations to it.
\end{remark}

\section{Integrable equations of Itoh--Narita--Bogoyavlensky type}
\label{secind}

Let $\diu=1$. 
Consider the Itoh--Narita--Bogoyavlensky (INB) equation~\cite{bogoy88,itoh75,narita82}
\begin{gather}
\lb{be}
\inb_t=\inb(\inb_2+\inb_1-\inb_{-1}-\inb_{-2})
\end{gather}
for a scalar function $\inb=\inb(n,t)$.
Fix a constant $\ap\in\mathbb{C}$.
Consider the modified INB equation
\begin{gather}
\lb{mbe}
u_t=u(u - \ap)(u_2u_1 - u_{-1}u_{-2})
\end{gather}
for a scalar function $u=u(n,t)$.
It is known (see, e.g.,~\cite{adler-pos14,bogoy89} and references therein) that, 
for any solution $u=u(n,t)$ of equation~\er{mbe}, the function
\begin{gather}
\lb{bembe}
\inb=u_2u_1(u_0-\ap)
\end{gather}
satisfies the INB equation~\er{be}. 
More precisely, formula~\er{bembe} determines a MT 
of order~$2$ from equation~\er{mbe} to equation~\er{be}.

Also, it is known that equation~\er{be} possesses the MLP
\begin{gather}
\lb{lpbe}
\mm^{\text{\tiny{INB}}}=\begin{pmatrix}
 0& 1& 0\\
 0& 0& 1\\
 -\inb_{0}& 0& \la\\
\end{pmatrix},\qquad
\lt^{\text{\tiny{INB}}}=\begin{pmatrix}
\inb_{-2} + \inb_{-1}& 0& \la \\
-\la \inb_{0}& \inb_{-1} + \inb_{0}& \la^2\\
-\la^2 \inb_{0}& -\la \inb_{1}& \la^3 + \inb_{0} + \inb_{1}\\
\end{pmatrix}.
\end{gather}
Substituting~\er{bembe} in~\er{lpbe}, one obtains the following MLP for equation~\er{mbe}
\begin{gather}
\lb{Mmbe}
\mm(u_{0},u_{1},u_{2},\la)=
\begin{pmatrix}
 0 & 1 & 0 \\
 0 & 0 & 1 \\
 (\ap-u_{0}) u_{1} u_{2} & 0 & \la \\
\end{pmatrix},\\
\lb{Umbe}
\begin{split}
\lt(& u_{-2},u_{-1},u_0,u_{1},u_{2},u_{3},\la)= \\
&=\left(\begin{smallmatrix}
 -u_{0} (\ap u_{-1}-u_{-2} u_{-1}-u_{1} u_{-1}+\ap u_{1}) & 0 & \la \\
 \la (\ap-u_{0}) u_{1} u_{2} & -u_{1} (\ap u_{0}-u_{-1} u_{0}-u_{2} u_{0}+\ap u_{2}) & \la^2 \\
 \la^2 (\ap-u_{0}) u_{1} u_{2} & \la (\ap-u_{1}) u_{2} u_{3} & \la^3-\ap u_{1} u_{2}+u_{0} u_{1} u_{2}-\ap u_{2} u_{3}+u_{1} u_{2} u_{3} \\
\end{smallmatrix}\right).
\end{split}
\end{gather}
The matrix-valued function~\er{Mmbe}  satisfies equations~\er{pdm},~\er{pdjq} 
with $k=q=2$, $\diu=1$, $u^1_\ell=u_\ell$, $\ell\in\zz$.
Therefore, this MLP is $2$-flat (in the sense of Definition~\ref{dsim}), 
and we can apply Theorem~\ref{thmuk} with $k=2$.

In order to use formula~\er{guum} in the case $k=2$,
we need to choose a constant $a_0\in\fik$, substitute $u_0=a_0$ in $\mm(u_{0},u_{1},u_{2},\la)$
given by~\er{Mmbe}, and consider the inverse matrix
\begin{gather}
\lb{inMmbe}
\big(\mm(a_0,u_1,u_2,\la)\big)^{-1}=
\begin{pmatrix}
 0 & 1 & 0 \\
 0 & 0 & 1 \\
 (\ap-a_{0}) u_{1} u_{2} & 0 & \la \\
\end{pmatrix}^{-1}=
\begin{pmatrix}
 0 & -\frac{\la }{(\ap-a_0) u_1 u_2} & \frac{1}{(\ap-a_0) u_1 u_2} \\
 1 & 0 & 0 \\
 0 & 1 & 0 \\
\end{pmatrix}.
\end{gather}
We can choose for $a_0$ any value such that \er{inMmbe} is well defined.
In order to make formula~\er{inMmbe} as simple as possible, we take $a_0=\ap+1$.

According to~\er{guum} for $k=2$ and $a_0=\ap+1$, we consider the gauge transformation
\begin{gather}
\lb{guu1m}
\mathbf{g}(u_0,u_1,\la)=\cS^{-1}\Big(\big(\mm(\ap+1,u_1,u_2,\la)\big)^{-1}\Big)=
\cS^{-1}\begin{pmatrix}
 0 & \frac{\la }{u_1 u_2} & \frac{-1}{u_1 u_2} \\
 1 & 0 & 0 \\
 0 & 1 & 0 \\
\end{pmatrix}=\begin{pmatrix}
 0 & \frac{\la }{u_0 u_1} & \frac{-1}{u_0 u_1} \\
 1 & 0 & 0 \\
 0 & 1 & 0 \\
\end{pmatrix}.
\end{gather}
Applying the gauge transformation~\er{guu1m} to the MLP~\er{Mmbe},~\er{Umbe}, 
one gets the following gauge equivalent MLP
\begin{gather}
\lb{tMmbe}
\tilde{\mm}(u_{0},u_{1},\la)=\cS(\mathbf{g})\cdot \mm\cdot\mathbf{g}^{-1}=
\begin{pmatrix}
 0 & u_{0}-\ap & 0 \\
 0 & 0 & 1 \\
 -u_{0} u_{1} & 0 & \la \\
\end{pmatrix},\\
\lb{tUmbe}
\begin{split}
&\tilde{\lt}(u_{-2},u_{-1},u_0,u_{1},u_{2},\la)=
\tdt(\mathbf{g})\cdot\mathbf{g}^{-1}+\mathbf{g}\cdot\lt\cdot\mathbf{g}^{-1}=\\
=&\left(\begin{smallmatrix}
 u_{-1} (-\ap u_{-2}-\ap u_{0}+u_{0} u_{-2}+u_{0} u_{1}) & 0 & -\la (\ap-u_{-1}) \\
 -\la u_{0} u_{1} & -u_{0} (\ap u_{-1}+\ap u_{1}-u_{-2} u_{-1}-u_{1} u_{-1}) & \la^2 \\
 -\la^2 u_{0} u_{1} & \la u_{1} u_{2} (\ap-u_{0}) & -\ap u_{0} u_{1}-\ap u_{1} u_{2}+\la^3+u_{-1} u_{0} u_{1}+u_{0} u_{1} u_{2} \\
\end{smallmatrix}\right).
\end{split}
\end{gather}
Here $\tdt$ is the total derivative operator corresponding to equation~\eqref{mbe}.
That is, for any function $\fn=\fn(\ud)$ depending on a finite number 
of the variables $u_\ell$ for $\ell\in\zz$, one has
\begin{gather}
\notag
\tdt(\fn)=
\sum_{\ell\in\mathbb{Z}}\cS^\ell\big(u_0(u_0-\ap)(u_2u_1 - u_{-1}u_{-2})\big)\cdot\frac{\pd \fn}{\pd u_\ell}.
\end{gather}

In agreement with the proof of Theorem~\ref{thmuk} in the case $k=2$,
since \er{tMmbe} is obtained from~\er{Mmbe} 
by means of the gauge transformation~\er{guu1m} constructed by formula~\er{guum}, 
the matrix-valued function~\er{tMmbe} depends only on~$u_{0},\,u_{1},\,\la$,
in contrast to the matrix-valued function~\er{Mmbe} depending on~$u_{0},\,u_{1},\,u_2,\,\la$.

The matrix-valued function~\er{tMmbe} satisfies equations~\er{pdm},~\er{pdjq} 
with $k=q=1$, $\diu=1$, $u^1_\ell=u_\ell$, $\ell\in\zz$.
Therefore, the MLP~\er{tMmbe},~\er{tUmbe} is $1$-flat, 
and we can apply Theorem~\ref{thmuk} with $k=1$ to~\er{tMmbe}.

Now we are going to use formula~\er{guum} in the case $k=1$ 
with $\mm$ replaced by~$\tilde{\mm}(u_{0},u_{1},\la)$ given by~\er{tMmbe}.
To this end, one needs to choose a constant $a_0\in\fik$, 
substitute $u_0=a_0$ in $\tilde{\mm}(u_{0},u_{1},\la)$, and consider the inverse matrix
\begin{gather}
\lb{intMmbe1}
\big(\tilde{\mm}(a_0,u_{1},\la)\big)^{-1}=
\begin{pmatrix}
 0 & a_0-\ap & 0 \\
 0 & 0 & 1 \\
 -a_0 u_{1} & 0 & \la \\
\end{pmatrix}^{-1}=
\begin{pmatrix}
 0 & \frac{\la}{a_0 u_{1}} & -\frac{1}{a_0 u_{1}} \\
 \frac{1}{a_0-\ap} & 0 & 0 \\
 0 & 1 & 0 \\
\end{pmatrix}.
\end{gather}
One can choose for $a_0$ any value such that \er{intMmbe1} is well defined.
Suppose that $\ap\neq 1$. (The case $\ap=1$ is discussed separately in Remark~\ref{ra1} below.)
Assuming $\ap\neq 1$, we take $a_0=1$.

Now we derive a gauge transformation $\tilde{\mathbf{g}}(u_0,\la)$ 
by means of formula~\er{guum} with $k=1$, $a_0=1$, $\tilde{\mm}(a_0,u_{1},\la)$.
That is, 
\begin{gather}
\lb{tguu1m}
\tilde{\mathbf{g}}(u_0,\la)=\cS^{-1}\Big(\big(\tilde{\mm}(1,u_{1},\la)\big)^{-1}\Big)=
\cS^{-1}\begin{pmatrix}
 0 & \frac{\la}{u_{1}} & -\frac{1}{u_{1}} \\
 \frac{1}{1-\ap} & 0 & 0 \\
 0 & 1 & 0 \\
\end{pmatrix}=\begin{pmatrix}
 0 & \frac{\la}{u_{0}} & -\frac{1}{u_{0}} \\
 \frac{1}{1-\ap} & 0 & 0 \\
 0 & 1 & 0 \\
\end{pmatrix}.
\end{gather}
Applying the gauge transformation~\er{tguu1m} 
to the MLP~\er{tMmbe},~\er{tUmbe}, we obtain the following gauge equivalent MLP
\begin{gather}
\lb{checkMGmbe}
\check{\mm}(u_{0},\la)=\cS(\tilde{\mathbf{g}})\cdot \tilde{\mm}\cdot\tilde{\mathbf{g}}^{-1}=
\begin{pmatrix}
 0 & -(\ap-1) u_{0} & 0 \\
 0 & 0 & \frac{\ap-u_{0}}{\ap-1} \\
 -u_{0} & 0 & \la \\
\end{pmatrix},\\
\lb{checkUGmbe}
\begin{split}
&\check{\lt}(u_{-2},u_{-1},u_0,u_{1},\la)=
\tdt(\tilde{\mathbf{g}})\cdot\tilde{\mathbf{g}}^{-1}+
\tilde{\mathbf{g}}\cdot\tilde{\lt}\cdot\tilde{\mathbf{g}}^{-1}=\\
=&\left(\begin{smallmatrix}
 -\ap u_{-2} u_{-1}-\ap u_{0} u_{1}+u_{-2} u_{0} u_{-1}+u_{0} u_{1} u_{-1} & 0 & -\la u_{-1} (\ap-u_{-2}) \\
 -\frac{\la u_{0} (\ap-u_{-1})}{\ap-1} & -u_{-1} (\ap u_{-2}+\ap u_{0}-u_{0} u_{-2}-u_{0} u_{1}) & \frac{\la^2 (\ap-u_{-1})}{\ap-1} \\
 -\la^2 u_{0} & (\ap-1) \la u_{0} u_{1} & -\ap u_{-1} u_{0}-\ap u_{0} u_{1}+\la^3+u_{-2} u_{-1} u_{0}+u_{-1} u_{0} u_{1} \\
 \end{smallmatrix}\right).
\end{split}
\end{gather}

\begin{remark}
\label{ra1}
In the above computation we assumed $\ap\neq 1$, which allowed us to take $a_0=1$.
If $\ap=1$ then one can take $a_0=2$ and proceed in a similar fashion.
\end{remark}

Thus, using the procedure described in Remark~\ref{rsgt} and Theorem~\ref{thmuk},
we have simplified the MLP~\er{Mmbe},~\er{Umbe} 
by means of gauge transformations and 
got the gauge equivalent MLP~\er{checkMGmbe},~\er{checkUGmbe}.
Its $\cS$-part~\er{checkMGmbe} depends only on~$u_{0},\,\la$,
in contrast to~\er{Mmbe} depending on~$u_{0},\,u_{1},\,u_{2},\,\la$.

As mentioned in Remark~\ref{rmcds}, 
the paper~\cite{BIg2016} presents a method to derive 
MTs from a MLP in the case when its $\cS$-part 
depends only on~$u_{0},\,\la$ and satisfies certain conditions.

It can be shown that this method of~\cite{BIg2016} is applicable 
to the obtained MLP~\er{checkMGmbe},~\er{checkUGmbe} for the 
modified INB equation~\er{mbe}.
In this way one can derive several modified versions of equation~\er{mbe} 
which are connected to~\er{mbe} by MTs.
In particular, for any fixed constant $\cb\in\mathbb{C}$ 
one can obtain the following modified version of~\er{mbe}
\begin{gather}
\lb{vtlong}
v_t=\frac{\ap^4 (\ap^2 v-\cb^2) (\ap+\cb v) (\ap v_{-1} v+\cb) (\ap v v_1+\cb)\,A(v_{-2},v_{-1},v_1,v_2)}{(\ap^3 v_{-2} v_{-1} v+\ap^3+\ap^2 \cb v_{-2}+\ap^2 \cb v_{-1}+\ap^2 \cb v-\cb^3)\,B(v_{-1},v,v_1,v_2)},\\
\notag
\begin{split}
A(&v_{-2},v_{-1},v_1,v_2)=\ap^3 v_{-2} v_{-1}-\ap^3 v_1 v_2+\ap^2 \cb v_{-2} v_{-1} v_1+\\
&+\ap^2 \cb v_{-2} v_{-1} v_2-\ap^2 \cb v_{-2} v_1 v_2-\ap^2 \cb v_{-1} v_1 v_2-\ap \cb^2 v_{-1}+\ap \cb^2 v_1+\ap \cb^2 v_2-\ap \cb^2 v_{-2}-\cb^3 v_{-2} v_{-1}+\cb^3 v_1 v_2,
\end{split}\\
\notag
\begin{split}
B(&v_{-1},v,v_1,v_2)=\\
&=(\ap^3 v_{-1} v v_1+\ap^3+\ap^2 \cb v_{-1}+\ap^2 \cb v+\ap^2 \cb v_1-\cb^3)
(\ap^3 v v_1 v_2+\ap^3+\ap^2 \cb v+\ap^2 \cb v_1+\ap^2 \cb v_2-\cb^3)
\end{split}
\end{gather}
such that equation~\er{vtlong} is connected to~\er{mbe} by the MT
\begin{gather}
\lb{dslong}
u=\frac{\ap (\ap v v_1+\cb) (\ap^2 v_2-\cb^2)}{\ap^3 v v_1 v_2+\ap^3+\ap^2 \cb v+\ap^2 \cb v_1+\ap^2 \cb v_2-\cb^3}.
\end{gather}
Substituting~\er{dslong} in~\er{checkMGmbe},~\er{checkUGmbe}, one 
obtains a $\la$-dependent matrix Lax pair for~\er{vtlong}. 
Therefore, equation~\er{vtlong} is integrable in the sense that 
\er{vtlong} possesses a nontrivial matrix Lax pair with parameter~$\la$
and is connected by the MT~\er{dslong} 
to the well-known integrable modified INB equation~\er{mbe}.

A derivation of~\er{vtlong} and~\er{dslong} from 
the MLP~\er{checkMGmbe},~\er{checkUGmbe}
by the method of~\cite{BIg2016} is discussed in Remark~\ref{rdermt} below.
Equation~\er{vtlong} is connected to the INB equation~\er{be} 
by the composition of the two MTs~\er{bembe},~\er{dslong}.
Therefore, \er{vtlong} can be regarded as a doubly modified version 
of the INB equation. 

Computing the composition of the MTs~\er{bembe},~\er{dslong}, 
we get the following MT of order~$4$
\begin{gather}
\lb{ds4}
\inb=-\frac{\ap^4 (\ap+\cb v_{0}) (\ap+\cb v_{1}) (\ap v_{1} v_{2}+\cb)%
 (\ap^2 v_{3}-\cb^2) (\ap v_{2} v_{3}+\cb)(\ap^2 v_{4}-\cb^2)}%
{(\ap^3 v_{0} v_{1} v_{2}+\ap^3+\ap^2 \cb v_{0}+\ap^2 \cb v_{1}+\ap^2 \cb v_{2}-\cb^3)\,%
H(v_{1},v_{2},v_{3},v_{4})},\\
\notag
\begin{split}
H(&v_{1},v_{2},v_{3},v_{4})=\\
&=(\ap^3 v_{1} v_{2} v_{3}+\ap^3+\ap^2 \cb v_{1}+\ap^2 \cb v_{2}+\ap^2 \cb v_{3}-\cb^3) 
(\ap^3 v_{2} v_{3} v_{4}+\ap^3+\ap^2 \cb v_{2}+\ap^2 \cb v_{3}+\ap^2 \cb v_{4}-\cb^3)
\end{split}
\end{gather}
from equation~\er{vtlong} to the INB equation~\er{be}.

\begin{remark}
\label{rGarYam}
MTs of order~$2$ over the INB equation~\er{be} were studied in~\cite{GarYam2019}.
The MT~\er{dslong} of order~$2$ does not appear in~\cite{GarYam2019}, 
since \er{dslong} is over the modified INB equation~\er{mbe}.

The MT~\er{ds4} is of order~$4$ 
and does not appear in~\cite{GarYam2019} either.

Some MTs of order~$2$ over equation~\er{mbe} are presented in~\cite{MikhXenit2014,xenit18},
but \er{dslong} does not appear in~\cite{MikhXenit2014,xenit18}.
\end{remark}

\begin{remark}
\label{rdermt}
According to Definition~\ref{dmlpgt}, we can consider the auxiliary linear system
\begin{gather}
\lb{sys2}
\begin{aligned}
\cS(\Phi)&=\check{\mm}(u_{0},\la)\Phi,\\
\frac{\pd}{\pd t}(\Phi)&=\check{\lt}(u_{-2},u_{-1},u_0,u_{1},\la)\Phi
\end{aligned}
\end{gather}
corresponding to the MLP~\er{checkMGmbe},~\er{checkUGmbe}.
Here $\Phi=\Phi(n,t,\la)$ is an invertible $2\times 2$ matrix-valued function.

Following ideas of~\cite{BIg2016}, consider the Taylor expansion with respect to~$\la$
\begin{gather}
\lb{tayl}
\Phi=\Phi^0+\la \Phi^1+O(\la^2),
\end{gather}
where $\Phi^0=\Phi^0(n,t)$ and $\Phi^1=\Phi^1(n,t)$ are $2\times 2$ matrix-valued functions,
and substitute~\er{tayl} in~\er{sys2}, which yields
\begin{gather}
\lb{sys21}
\begin{aligned}
\cS(\Phi^0)+\la\cS(\Phi^1)&=\check{\mm}(u_{0},\la)\cdot(\Phi^0+\la \Phi^1)\qquad\text{modulo $O(\la^2)$},\\
\frac{\pd}{\pd t}(\Phi^0)+\la\frac{\pd}{\pd t}(\Phi^1)
&=\check{\lt}(u_{-2},u_{-1},u_0,u_{1},\la)\cdot(\Phi^0+\la \Phi^1)\qquad\text{modulo $O(\la^2)$}.
\end{aligned}
\end{gather}

Using the method of~\cite{BIg2016}, one can derive equation~\er{vtlong} and the MT~\er{dslong}
as a certain reduction of the matrix system~\er{sys21}.
A detailed derivation of~\er{vtlong} and~\er{dslong} from~\er{sys21}
by the method of~\cite{BIg2016} will be presented in a different publication.
\end{remark}

\section{Integrable equations related to the Toda lattice}
\label{seceqtod}

Consider the Toda lattice equation 
\begin{gather}
\lb{toda}
\gamma_{tt}=\exp(\gamma_{1}-\gamma)-\exp(\gamma-\gamma_{-1}),\qquad\quad
\gamma_1=\gamma(n+1,t),\qquad \gamma_{-1}=\gamma(n-1,t),
\end{gather}
for a scalar function $\gamma=\gamma(n,t)$.
Following~\cite{flaschka74,manakov75}, we consider 
the functions $\td^1(n,t)=\exp(\gamma-\gamma_{-1})$ and $\td^2(n,t)=\gamma_{t}$.
Then~\eqref{toda} implies  
\begin{gather}
\label{todanew}
\left\{
\begin{aligned}
\td^1_t&=\td^1(\td^2-\td^2_{-1}),\\
\td^2_t&=\td^1_1-\td^1.
\end{aligned}\right.
\end{gather}
The $2$-component equation~\eqref{todanew} 
is sometimes called the Toda lattice written in the Flaschka--Manakov coordinates.

It is known (see, e.g.,~\cite{kmw}) that the following matrices form a MLP for~\eqref{todanew}
\begin{gather}
\lb{lrnew}
\mm= \begin{pmatrix} 
\la+\td^2 & \td^1 \\ -1 & 0 
\end{pmatrix},\qquad\quad
\lt=  
\begin{pmatrix}
0 &-\td^1\\
1&\la+\td^2_{-1}
\end{pmatrix}.
\end{gather}
Using known methods to derive MTs from a given MLP
(see, e.g.,~\cite{BIg2016} and references therein), 
from the MLP~\er{lrnew} one can obtain the following.
Fix a constant $\ct\in\mathbb{C}$. The $2$-component equation
\begin{gather}
\lb{mtoda}
\left\{
\begin{gathered}
u^{1}_t=
\frac{u^{1}_{0}\, H^1(u^{1}_{-1},u^{2}_{-1},u^{1}_{0},u^{2}_{0},u^{1}_{1},u^{2}_{1},\ct)}{(u^{1}_{-1}-u^{2}_{-1}) (u^{1}_{0}-u^{2}_{0})},\\
u^{2}_t=
-\frac{u^{2}_{0}\, H^2(u^{1}_{-1},u^{2}_{-1},u^{1}_{0},u^{2}_{0},u^{1}_{1},u^{2}_{1},\ct)}{(u^{1}_{-1}-u^{2}_{-1}) (u^{2}_{0}-u^{1}_{0})},\\
\end{gathered}\right.
\end{gather}
where
\begin{gather*}
\begin{split}
H^1(u^{1}_{-1},u^{2}_{-1},u^{1}_{0},u^{2}_{0},&u^{1}_{1},u^{2}_{1},\ct)=
-\ct u^{2}_{-1} u^{1}_{0}+\ct u^{1}_{-1} u^{2}_{0}+u^{2}_{-1} (u^{1}_{0})^2-2 u^{2}_{-1} u^{2}_{0} u^{1}_{0}+\\
&+u^{2}_{-1} (u^{2}_{0})^2-u^{1}_{-1} u^{1}_{1} u^{2}_{0}+u^{1}_{1} u^{2}_{-1} u^{2}_{0}+u^{1}_{-1} u^{2}_{0} u^{2}_{1}-u^{2}_{-1} u^{2}_{0} u^{2}_{1},
\end{split}\\
\begin{split}
H^2(u^{1}_{-1},u^{2}_{-1},u^{1}_{0},u^{2}_{0},&u^{1}_{1},u^{2}_{1},\ct)=
-\ct u^{2}_{-1} u^{1}_{0}+\ct u^{1}_{-1} u^{2}_{0}+u^{1}_{-1} (u^{1}_{0})^2-u^{1}_{-1} u^{1}_{1} u^{1}_{0}+\\
&+u^{1}_{1} u^{2}_{-1} u^{1}_{0}-2 u^{1}_{-1} u^{2}_{0} u^{1}_{0}+u^{1}_{-1} u^{2}_{1} u^{1}_{0}-u^{2}_{-1} u^{2}_{1} u^{1}_{0}+u^{1}_{-1} (u^{2}_{0})^2,
\end{split}
\end{gather*}
is connected to~\er{todanew} by the MT
\begin{gather}
\lb{dstoda}
\left\{
\begin{aligned}
\td^1&=
\frac{u^{1}_{0} u^{2}_{0} (u^{1}_{1}-u^{2}_{1}-\ct)}{u^{1}_{0}-u^{2}_{0}},\\
\td^2&=
\frac{\ct u^{2}_{0}-u^{1}_{1} u^{1}_{0}+u^{2}_{0} u^{2}_{1}}{u^{1}_{0}-u^{2}_{0}}.\\
\end{aligned}\right.
\end{gather}

In this section we use the notation~\er{uldiu} with $\diu=2$.
That is, for each $\ell\in\zz$ we have $u_\ell=(u^1_\ell,u^2_\ell)$.
Substituting~\er{dstoda} in~\er{lrnew}, one obtains the following MLP for equation~\er{mtoda}
\begin{gather}
\lb{Mmtoda}
\mm(u_{0},u_{1},\la)=\begin{pmatrix} 
\la+
\frac{\ct u^{2}_{0}-u^{1}_{1} u^{1}_{0}+u^{2}_{0} u^{2}_{1}}{u^{1}_{0}-u^{2}_{0}} & 
\frac{u^{1}_{0} u^{2}_{0} (u^{1}_{1}-u^{2}_{1}-\ct)}{u^{1}_{0}-u^{2}_{0}} \\ 
-1 & 0 
\end{pmatrix},\\
\lb{Umtoda}
\lt(u_{-1},u_{0},u_{1},\la)=
\begin{pmatrix}
0 &\frac{u^{1}_{0} u^{2}_{0} (u^{2}_{1}+\ct-u^{1}_{1})}{u^{1}_{0}-u^{2}_{0}}\\
1&\la+\frac{\ct u^{2}_{-1}-u^{1}_{0} u^{1}_{-1}+u^{2}_{-1} u^{2}_{0}}{u^{1}_{-1}-u^{2}_{-1}}
\end{pmatrix}.
\end{gather}
The matrix-valued function~\er{Mmtoda}  satisfies equations~\er{pdm},~\er{pdjq} 
with $k=q=1$, $\diu=2$.
Therefore, this MLP is $1$-flat (in the sense of Definition~\ref{dsim}), 
and we can apply Theorem~\ref{thmuk} with $k=1$.

In order to use formula~\er{guum} in the case $k=1$,
we need to choose a constant vector $a_0=(a^1_0,a^2_0)\in\fik^2$, 
substitute $u^1_0=a^1_0$, $u^2_0=a^2_0$ in $\mm(u_{0},u_{1},\la)$
given by~\er{Mmtoda}, and consider the inverse matrix
\begin{gather}
\lb{inMtoda}
\big(\mm(a_0,u_1,\la)\big)^{-1}=
\begin{pmatrix} 
\la+
\frac{\ct a^{2}_{0}-u^{1}_{1} a^{1}_{0}+a^{2}_{0} u^{2}_{1}}{a^{1}_{0}-a^{2}_{0}} & 
\frac{a^{1}_{0} a^{2}_{0} (u^{1}_{1}-u^{2}_{1}-\ct)}{a^{1}_{0}-a^{2}_{0}} \\ 
-1 & 0 
\end{pmatrix}^{-1}=
\begin{pmatrix}
 0 & -1 \\
 \frac{a^{2}_{0}-a^{1}_{0}}{a^{1}_{0} a^{2}_{0} (\ct-u^{1}_{1}+u^{2}_{1})} & 
\frac{\la(a^{2}_{0}-a^{1}_{0})+u^{1}_{1} a^{1}_{0}-\ct a^{2}_{0}-a^{2}_{0} u^{2}_{1}}{a^{1}_{0} a^{2}_{0} (\ct-u^{1}_{1}+u^{2}_{1})} \\
\end{pmatrix}.
\end{gather}
One can choose for $a_0=(a^1_0,a^2_0)$ any value such that \er{inMtoda} is well defined.
We take $a^1_0=2$, $\,a^2_0=1$.

According to~\er{guum} for $k=1$ and $a_0=(2,1)$, we consider the gauge transformation
\begin{gather}
\lb{gmtoda}
\mathbf{g}(u_0,\la)=\cS^{-1}\Big(\big(\mm(a_0,u_1,\la)\big)^{-1}\Big)=
\cS^{-1}\begin{pmatrix}
 0 & -1 \\
 \frac{-1}{2 (\ct-u^{1}_{1}+u^{2}_{1})} & \frac{-\la+2 u^{1}_{1}-\ct-u^{2}_{1}}{2 (\ct-u^{1}_{1}+u^{2}_{1})} \\
\end{pmatrix}
=\begin{pmatrix}
 0 & -1 \\
 \frac{-1}{2 (\ct-u^{1}_{0}+u^{2}_{0})} & \frac{-\la+2 u^{1}_{0}-\ct-u^{2}_{0}}{2 (\ct-u^{1}_{0}+u^{2}_{0})} \\
\end{pmatrix}.
\end{gather}
Below we use the total derivative operator~$\tdt$ corresponding to equation~\eqref{mtoda}.
Applying the gauge transformation~\er{gmtoda} to the MLP~\er{Mmtoda},~\er{Umtoda}, 
one gets the following gauge equivalent MLP
\begin{gather}
\lb{tMmtoda}
\tilde{\mm}(u_{0},\la)=\cS(\mathbf{g})\cdot \mm\cdot\mathbf{g}^{-1}=
\begin{pmatrix}
 \la+\ct-2 u^{1}_{0}+u^{2}_{0} & -2 (\ct-u^{1}_{0}+u^{2}_{0}) \\
 \frac{1}{2} (\la+\ct-2 u^{1}_{0}+2 u^{2}_{0}-\frac{(\la+\ct) u^{2}_{0}}{u^{1}_{0}-u^{2}_{0}}) & \frac{(u^{1}_{0}-2 u^{2}_{0}) (-\ct+u^{1}_{0}-u^{2}_{0})}{u^{1}_{0}-u^{2}_{0}} \\
\end{pmatrix},\\
\lb{tUmtoda}
\begin{split}
&\tilde{\lt}(u_{-1},u_0,\la)=
\tdt(\mathbf{g})\cdot\mathbf{g}^{-1}+\mathbf{g}\cdot\lt\cdot\mathbf{g}^{-1}=\\
=&\begin{pmatrix}
 -\frac{(u^{1}_{-1}-2 u^{2}_{-1}) (\ct-u^{1}_{0}+u^{2}_{0})}{u^{1}_{-1}-u^{2}_{-1}} & 2 (\ct-u^{1}_{0}+u^{2}_{0}) \\
u^{1}_{0}-u^{2}_{0}+
 \frac{(\la+\ct)(2u^{2}_{-1}-u^{1}_{-1})}{2 (u^{1}_{-1}-u^{2}_{-1})} & \la+\ct-\frac{(2 u^{1}_{-1}-u^{2}_{-1}) (u^{1}_{0}-u^{2}_{0})}{u^{1}_{-1}-u^{2}_{-1}} 
\end{pmatrix}.
\end{split}
\end{gather}
Here, applying the procedure discussed in Remark~\ref{rsgt},
we have simplified the MLP~\er{Mmtoda},~\er{Umtoda} 
by the gauge transformation~\er{gmtoda} and 
obtained the gauge equivalent MLP~\er{tMmtoda},~\er{tUmtoda}.
The $\cS$-part~\er{tMmtoda} of the obtained MLP depends only on~$u_{0},\,\la$,
unlike the matrix~\er{Mmtoda} depending on~$u_{0},\,u_{1},\,\la$.

As discussed in Remark~\ref{rmcds}, 
the paper~\cite{BIg2016} describes a method to derive 
MTs from a MLP with $\cS$-part depending only on~$u_{0},\,\la$ 
and satisfying certain conditions.
This method is applicable 
to the obtained MLP~\er{tMmtoda},~\er{tUmtoda} for equation~\er{mtoda}, 
where $\ct$ is a fixed constant.
In this way one can derive several modified versions of equation~\er{mtoda} 
which are connected to~\er{mtoda} by MTs.
Since in the case of arbitrary~$\ct$ the corresponding formulas are rather cumbersome,
below we consider the case $\ct=0$.

The matrix-valued functions~\er{tMmtoda},~\er{tUmtoda} with $\ct=0$ 
form a MLP for equation~\er{mtoda} with $\ct=0$.
From this MLP one can derive the following. The $2$-component equation
\begin{gather}
\lb{longtoda}
\left\{
\begin{aligned}
&\wv^1_t=
-\frac{\wv^{1}_{0} (\wv^{1}_{0}+1) (\wv^{1}_{0} \wv^{1}_{1}+\wv^{2}_{1} \wv^{1}_{1}+\wv^{1}_{1}-\wv^{1}_{0} \wv^{2}_{1})}{\wv^{1}_{0} \wv^{1}_{1} \wv^{2}_{0}+\wv^{1}_{1} \wv^{2}_{0}-\wv^{1}_{0} \wv^{2}_{1} \wv^{2}_{0}+\wv^{1}_{1} \wv^{2}_{1} \wv^{2}_{0}-\wv^{1}_{0} \wv^{2}_{1}-\wv^{1}_{0} \wv^{1}_{1} \wv^{2}_{1}},\\
&\wv^2_t=
\frac{P(\wv^{1}_{-1},\wv^{1}_{0},\wv^{2}_{0},\wv^{1}_{1},\wv^{2}_{1})}{(\wv^{1}_{-1}-\wv^{1}_{0}) (\wv^{1}_{0} \wv^{1}_{1} \wv^{2}_{0}+\wv^{1}_{1} \wv^{2}_{0}-\wv^{1}_{0} \wv^{2}_{1} \wv^{2}_{0}+\wv^{1}_{1} \wv^{2}_{1} \wv^{2}_{0}-\wv^{1}_{0} \wv^{2}_{1}-\wv^{1}_{0} \wv^{1}_{1} \wv^{2}_{1})},\end{aligned}\right.
\end{gather}
where 
\begin{gather*}
\wv^1=\wv^1(n,t),\qquad\quad\wv^2=\wv^2(n,t),\\
P(\wv^{1}_{-1},\wv^{1}_{0},\wv^{2}_{0},\wv^{1}_{1},\wv^{2}_{1})=
-\wv^{1}_{-1} \wv^{1}_{1} \wv^{2}_{0} (\wv^{1}_{0})^2+\wv^{1}_{-1} \wv^{2}_{1} (\wv^{1}_{0})^2+\wv^{1}_{-1} \wv^{1}_{1} \wv^{2}_{1} (\wv^{1}_{0})^2+\\
+\wv^{1}_{1} \wv^{2}_{1} (\wv^{1}_{0})^2+\wv^{1}_{-1} \wv^{2}_{0} \wv^{2}_{1} (\wv^{1}_{0})^2+\wv^{1}_{1} \wv^{2}_{0} \wv^{2}_{1} (\wv^{1}_{0})^2+\wv^{2}_{0} \wv^{2}_{1} (\wv^{1}_{0})^2+\\
+\wv^{2}_{1} (\wv^{1}_{0})^2-2 \wv^{1}_{-1} \wv^{1}_{1} \wv^{2}_{0} \wv^{1}_{0}-2 \wv^{1}_{-1} \wv^{1}_{1} \wv^{2}_{0} \wv^{2}_{1} \wv^{1}_{0}-\wv^{1}_{-1} \wv^{1}_{1} \wv^{2}_{0}-\wv^{1}_{-1} \wv^{1}_{1} \wv^{2}_{0} \wv^{2}_{1},
\end{gather*}
is connected to equation~\er{mtoda} with $\ct=0$ by the MT
\begin{gather}
\lb{dsltoda}
\left\{
\begin{gathered}
u^{1}=
\frac{\wv^{1}_{0} (\wv^{1}_{1}+1) (\wv^{1}_{0} \wv^{1}_{1}+\wv^{2}_{1} \wv^{1}_{1}+\wv^{1}_{1}-\wv^{1}_{0} \wv^{2}_{1})}{(\wv^{1}_{1}-\wv^{1}_{0})(-\wv^{1}_{1} \wv^{2}_{0}-\wv^{1}_{0} \wv^{1}_{1} \wv^{2}_{0}+\wv^{1}_{0} \wv^{2}_{1} \wv^{2}_{0}-\wv^{1}_{1} \wv^{2}_{1} \wv^{2}_{0}+\wv^{1}_{0} \wv^{2}_{1}+\wv^{1}_{0} \wv^{1}_{1} \wv^{2}_{1})},\\
u^{2}=
\frac{(\wv^{1}_{0}+1) \wv^{1}_{1} (\wv^{1}_{0} \wv^{1}_{1}+\wv^{2}_{1} \wv^{1}_{1}+\wv^{1}_{1}-\wv^{1}_{0} \wv^{2}_{1})}{(\wv^{1}_{1}-\wv^{1}_{0})(-\wv^{1}_{1} \wv^{2}_{0}-\wv^{1}_{0} \wv^{1}_{1} \wv^{2}_{0}+\wv^{1}_{0} \wv^{2}_{1} \wv^{2}_{0}-\wv^{1}_{1} \wv^{2}_{1} \wv^{2}_{0}+\wv^{1}_{0} \wv^{2}_{1}+\wv^{1}_{0} \wv^{1}_{1} \wv^{2}_{1})}.
\end{gathered}\right.
\end{gather}
Hence equation~\er{longtoda}
can be viewed as a modified version of equation~\er{mtoda} with $\ct=0$.

Equation~\er{longtoda} is connected to equation~\er{todanew}
by the composition of the two MTs~\er{dsltoda} and~\er{dstoda} with $\ct=0$.
Therefore, \er{longtoda} can be regarded as a doubly modified version of equation~\er{todanew}. 

Substituting~\er{dsltoda} in the matrix-valued functions~\er{tMmtoda},~\er{tUmtoda} with $\ct=0$, 
we derive a $\la$-dependent matrix Lax pair for~\er{longtoda}. 
Hence equation~\er{longtoda} is integrable in the sense that it 
has a nontrivial matrix Lax pair with parameter~$\la$ 
and is connected to the known integrable equation~\er{todanew}
by the composition of the two MTs discussed above.

A detailed derivation of~\er{longtoda} and~\er{dsltoda} 
from the MLP~\er{tMmtoda},~\er{tUmtoda} with $\ct=0$
by the method of~\cite{BIg2016} will be presented in a different publication.

\section{Conclusions}
\label{sconcl}

In this paper we have presented sufficient conditions 
which allow one to simplify a given MLP by gauge transformations 
and described a procedure for such a simplification.
Furthermore, we have demonstrated how this procedure helps 
(in the considered examples) to construct modified and doubly modified versions 
with MTs for a given differential-difference equation possessing a MLP satisfying certain conditions.

This has allowed us to derive new integrable equations connected by new MTs to known equations.
In particular, we have obtained new integrable equations~\er{vtlong}, \er{mtoda}, \er{longtoda}
and new MTs \er{dslong}, \er{ds4}, \er{dstoda}, \er{dsltoda}.

The results of this paper considerably extend the ones of~\cite{BIg2016},
since the paper~\cite{BIg2016} 
(where MLPs are called Darboux-Lax representations)
does not contain any analog of our Theorem~\ref{thmuk}.
Furthermore, the obtained integrable equations~\er{vtlong}, \er{mtoda}, \er{longtoda}
and MTs \er{dslong}, \er{ds4}, \er{dstoda}, \er{dsltoda} do not appear in~\cite{BIg2016}.

Motivated by the results of the present paper, we suggest the following open problems
for further research.

\textbf{Open problem 1.}
 As shown in Sections~\ref{secsimpl},~\ref{secind},~\ref{seceqtod}, 
	our results on simplifications of differential-difference matrix Lax pairs 
	(satisfying certain conditions) 
	by gauge transformations are based on a detailed study 
	of the action of the gauge transformations group
	on the set of MLPs of a given differential-difference equation 
	and allow us to construct new MTs for some differential-difference equations.
It would be interesting to obtain analogous results 
in the case of (1+1)-dimensional partial differential equations (PDEs).
In this direction one can use results of~\cite{igonarx12,igmanno19}
about the gauge transformations group action on matrix Lax pairs 
(also called zero-curvature representations) for (1+1)-dimensional PDEs
and constructions of MTs from some zero-curvature representations 
of such PDEs~\cite{drin-sok85,igon05}.

\textbf{Open problem 2.} 
In this paper we consider functions depending on dynamical variables 
(e.g. $u_\ell$ or $v_\ell$, $\,\ell\in\zz$,) and, possibly, on some parameters.
This means that we study the \emph{autonomous case}.
It would be interesting to extend the obtained results to the \emph{non-autonomous case},
where (in addition to dependence on dynamical variables and parameters) 
explicit dependence on the independent variables $n,t$ is allowed.

\section*{Acknowledgments}
The author would like to thank Rustem Garifullin and Maxim Pavlov for useful discussions.

The results in Section~\ref{secsimpl} 
were obtained with support from the grant No.~{21-71-30011}
of the Russian Science Foundation, \url{https://rscf.ru/en/project/21-71-30011/}.

The research presented in Sections~\ref{secind},~\ref{seceqtod}
was funded by the Ministry of Science and Higher Education of Russia
(Agreement on allocation of subsidy No.~{075-02-2024-1442}
for the Regional Mathematical Center ``Center of Integrable Systems'' 
of Yaroslavl State University).

\end{document}